\documentclass[12pt,a4paper]{article}
\usepackage{}

\usepackage{latexsym}
\usepackage{amsmath}
\usepackage{amssymb}
\usepackage{arydshln}

\usepackage{cite}

\usepackage{stmaryrd}
\usepackage{enumerate}

\usepackage{hyperref}

\usepackage{color}
\usepackage{lineno}
\usepackage{graphicx}
\usepackage{ae}
\usepackage{amsmath}
\usepackage{amssymb}
\usepackage{latexsym}
\usepackage{url}
\usepackage{epsfig}
\usepackage{mathrsfs}
\usepackage{amsfonts}
\usepackage{amsthm}

\newtheorem{theorem}{Theorem}[section]

\newtheorem{lemma}[theorem]{Lemma}

\newtheorem{cor}[theorem]{Corollary}

\newtheorem{example}{Example}

\theoremstyle{definition}
\newtheorem{definition}{Definition}
\newtheorem{claim}{Claim}

\numberwithin{equation}{section} 

\allowdisplaybreaks    

\def\qed{\hfill$\Box$\vspace{12pt}}

\long\def\delete#1{}

\usepackage{xcolor}
\usepackage[normalem]{ulem}

\usepackage{tikz}
\usetikzlibrary{decorations.markings}
\tikzstyle{vertex}=[circle, inner sep=1.2pt, minimum size=3pt]
\tikzstyle{filledvertex}=[circle, draw, fill, inner sep=1.2pt, minimum size=3pt]

\newcommand{\vertex}{\node[vertex]}

\tikzstyle{directed}=[postaction={decorate,
	decoration={markings,mark=at position 0.5 with {\arrow{stealth}}}
}]

\begin{document}
\title{State transfers in vertex complemented coronas}

\date{}
\author{~Jing Wang$^{a,b}$, ~Xiaogang Liu$^{a,b,c,}$\thanks{Supported by the National Natural Science Foundation of China (No. 11601431), the Natural Science Foundation of Shaanxi Province (No. 2020JM-099) and the Natural Science Foundation of Qinghai Province  (No. 2020-ZJ-920).}~$^,$\thanks{ Corresponding author. Email addresses: wj66@mail.nwpu.edu.cn, xiaogliu@nwpu.edu.cn}
\\[2mm]
{\small $^a$School of Mathematics and Statistics,}\\[-0.8ex]
{\small Northwestern Polytechnical University, Xi'an, Shaanxi 710072, P.R.~China}\\
{\small $^b$Xi'an-Budapest Joint Research Center for Combinatorics,}\\[-0.8ex]
{\small Northwestern Polytechnical University, Xi'an, Shaanxi 710129, P.R. China}\\
{\small $^c$School of Mathematics and Statistics,}\\[-0.8ex]
{\small Qinghai Nationalities University, Xining, Qinghai 810007, P.R. China}
}
\date{}

\openup 0.5\jot
\maketitle

\begin{abstract}
In this paper, we study the existence of perfect state transfer and pretty good state transfer in vertex complemented coronas. We prove that perfect state transfer in vertex complemented coronas is extremely rare. In contrast, we give sufficient conditions for vertex complemented coronas to have pretty good state transfer.

\smallskip



\end{abstract}

\section{Introduction}
Let $G$ be a graph  with adjacency matrix $A_G$. The \emph{transition matrix} \cite{FarhiG98} of $G$ with respect to $A_G$ is defined by
$$
H_{A_{G}}(t) = \exp(-\mathrm{i}tA_{G})=\sum_{k=0}^{\infty}\frac{(-\mathrm{i})^{k} A^{k}_{G} t^{k}}{k!}, ~ t \in \mathbb{R},~\mathrm{i}=\sqrt{-1}.
$$
Let $H_{A_G}(t)_{u,v}$ denote the $(u,v)$-entry of $H_{A_G}(t)$, where $u,v\in V(G)$. If $u$ and $v$ are distinct vertices in $G$ and there is a time $\tau$ such that
$$
|H_{A_G}(\tau)_{u,v}|=1,
$$
then we say that \emph{perfect state transfer} (PST for short) from $u$ to $v$ occurs at time $\tau$ \cite{Bose03}. In particular, if $|H_{A_G}(\tau)_{u,u}|=1$, then we say that $G$ is \emph{periodic} relative to the vertex $u$ at time $\tau$ or $u$ is a \emph{periodic vertex} of $G$ at time $\tau$ \cite{Godsil11}. If every vertex of $G$ is periodic at the same time $\tau$, then $G$ is called a \emph{periodic graph} with the \emph{period} $\tau$  \cite{Godsil11}.

It is known \cite{Bose03} that PST is very important in quantum computing and quantum information processing. However, determining all graphs that admit PST is substantially difficult. In 2012, Godsil \cite[Corollary~6.2]{Godsil12} showed that there are at most finitely many connected graphs with a given maximum valency where PST occurs. Thus, Godsil posed to study a relaxation of PST, \emph{pretty good state transfer} (PGST for short) \cite{CGodsil}. A graph $G$ is said to have PGST from vertex $u$ to vertex $v$ \cite{CGodsil} if for each $\varepsilon > 0$, there exists a time $\tau$ such that
$$
\mid H_{A_G}(\tau)_{u,v}\mid \geq 1-\varepsilon.
$$

Up until now, many graphs have been proved to have or not have PST as well as PGST, including trees \cite{Bose03, CoutinhoL2015, Fan, GodsilKSS12}, Cayley graphs \cite{Basic11, Basic09, CaoCL20, CaoF21, CaoWF20, CC, LiLZZ21, HPal3, HPal, HPal5, Tan19, Tm19}, distance regular graphs \cite{Coutinho15} and some graph operations such as NEPS \cite{chris1, chris2,  LiLZZ21, HPal1, HPal4, SZ}, coronas \cite{AckBCMT16} and joins \cite{Angeles10}.
For more information, we refer the reader to \cite{Coh14, Coh19, CGodsil, Godsil12, HZ, Zhou14}.

In this paper, we investigate the existence of PST and PGST in a new graph operation, the so-called \emph{vertex complemented corona}, whose definition is given in Definition \ref{definition}.

\begin{definition}\label{definition}
Let $G$ be a graph with  vertex set $V(G)=\{v_1,v_2,\ldots,v_n\}$ and let $\overrightarrow{H}=(H_1,H_2,\ldots,H_n)$
be an $n$-tuple of graphs. The \emph{vertex complemented corona}
$G\tilde{\circ}\overrightarrow{H}$ is formed by taking the
disjoint union of $G$ and $H_1,\ldots,H_n$ with each $H_i$ corresponding to the vertex $v_i$, and then joining every vertex in $H_{i}$ to every vertex in $V(G)\setminus\{v_{i}\}$ for $i=1,2,\ldots,n$.
\end{definition}

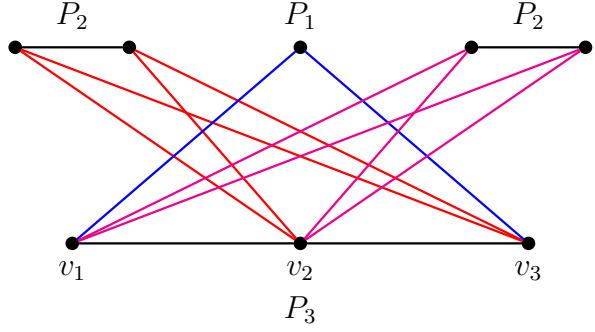
\begin{figure}
\begin{center}
\begin{tikzpicture}[x=0.75cm, y=0.65cm]
\tikzstyle{vertex}=[circle,inner sep=1.8pt, minimum size=0.1pt]

\vertex (a)[fill] at (0,0)[label=below:$v_{1}$]{};
\vertex (b)[fill] at (4,0)[label=below:$v_{2}$]{};
\vertex (c)[fill] at (8,0)[label=below:$v_{3}$]{};
\vertex (g1) at (4,-2)[label=above:$P_{3}$]{};
\draw[line width=.3mm,line cap=round](a)--(b);
\draw[line width=.3mm,line cap=round](b)--(c);

\vertex (d)[fill] at (-1,4){};
\vertex (e)[fill] at (1,4){};
\draw[line width=.3mm,line cap=round](d)--(e);
\vertex (g2) at (0,4)[label=above:$P_{2}$]{};
\draw[line width=.3mm,line cap=round,color=red](d)--(b);
\draw[line width=.3mm,line cap=round,color=red](d)--(c);
\draw[line width=.3mm,line cap=round,color=red](e)--(b);
\draw[line width=.3mm,line cap=round,color=red](e)--(c);

\vertex (f)[fill] at (4,4){};
\vertex (g3) at (4,4)[label=above:$P_{1}$]{};

\draw[line width=.3mm,line cap=round,color=blue](f)--(a);
\draw[line width=.3mm,line cap=round,color=blue](f)--(c);

\vertex (g)[fill] at (7,4){};
\vertex (h)[fill] at (9,4){};
\draw[line width=.3mm,line cap=round](g)--(h);
\vertex (g4) at (8,4)[label=above:$P_{2}$]{};
\draw[line width=.3mm,line cap=round,color=magenta](g)--(a);
\draw[line width=.3mm,line cap=round,color=magenta](g)--(b);
\draw[line width=.3mm,line cap=round,color=magenta](h)--(a);
\draw[line width=.3mm,line cap=round,color=magenta](h)--(b);

\end{tikzpicture}
\end{center}
\vspace{-0.8cm}
\caption{An example of the vertex complemented corona}
\label{VCC-Fig-1}
\end{figure}

Figure \ref{VCC-Fig-1} depicts the vertex complemented corona $P_3\tilde{\circ}\overrightarrow{H}$ with $\overrightarrow{H}=(P_2,P_1,P_2)$, where $P_n$ denotes the path on $n$ vertices.

In our work, we first compute eigenvalues and eigenprojectors of vertex complemented coronas. Then, we prove that PST in vertex complemented coronas is extremely rare by verifying there is no periodic vertex in vertex complemented coronas. In contrast, we give some sufficient conditions for vertex complemented coronas to have PGST.

\section{Preliminaries}

In this section, we list some basic results and notations, which will be useful for our paper.

\begin{lemma}\label{suhur complemented}\emph{(see \cite{Zhang})}
Let $M_{1}$, $M_{2}$, $M_{3}$ and $M_{4}$ be respectively $p\times p$, $p\times q$, $q\times p$ and $q\times q$ matrices with $M_{1}$ and $M_{4}$ invertible. Then
     \begin{align*}
  \det\left(
        \begin{array}{cc}
          M_{1} & M_{2} \\
          M_{3} & M_{4} \\
        \end{array}
      \right) &=\det(M_{4})\cdot\det(M_{1}-M_{2}M_{4}^{-1}M_{3}) \\
              &=\det(M_{1})\cdot\det(M_{4}-M_{3}M_{1}^{-1}M_{2}),
      \end{align*}
where $M_{1}-M_{2}M_{4}^{-1}M_{3}$ and $M_{4}-M_{3}M_{1}^{-1}M_{2}$ are called the Schur complements of $M_{4}$ and $M_{1}$, respectively.
\end{lemma}

The $M$-coronal $\Gamma_M(x)$ of an $n\times n$ matrix $M$ \cite{Cui-Tian12, kn:McLeman11} is defined to be the sum of the entries of the matrix $(xI_n-M)^{-1}$, that is,
$$\Gamma_M(x)=\mathbf{j}_n^\top(xI_n-M)^{-1}\mathbf{j}_n,$$
where $\mathbf{j}_n$ denotes the column vector of size $n$ with all entries equal to one, and $\mathbf{j}_n^\top$ denotes the transpose of $\mathbf{j}_n$.

\begin{lemma}\label{gamma}\emph{(see \cite[Proposition 2]{Cui-Tian12})}
If $M$ is an $n\times n$ matrix with each row sum equal to a constant $t$, then
\begin{equation*}
  \Gamma_{M}(x)=\frac{n}{x-t}.
\end{equation*}
\end{lemma}

\begin{lemma}\label{liuzhang1}\emph{(see \cite[Corollary 2.3]{Liu-zhang19})}
Let $\alpha$ be a real number, $A$ an $n\times n$ real matrix, $I_n$ the identity matrix of size $n$, and $J_{n}$ the $n\times n$ matrix with all entries equal to one. Then
\begin{equation*}
  \det(xI_{n}-A-\alpha J_{n})=(1-\alpha \Gamma_{A}(x))\det(xI_{n}-A).
\end{equation*}
\end{lemma}

We will need the Kronecker's Approximation Theorem to study the existence of PGST in vertex complemented coronas.
\begin{theorem}\label{hardy-wright}
\emph{(see \cite[Theorem 442]{Hw})}
\label{H-W}
Let $1,\lambda_1,\lambda_2,\ldots,\lambda_m$ be linearly independent over $\mathbb{Q}$. Let $\alpha_1,\alpha_2,\ldots,\alpha_m$ be arbitrary real numbers, and let $\varepsilon$ be a positive real number. Then there exist  integers $l$ and  $q_1,q_2,\ldots, q_m$ such that
\begin{equation}
\label{KroApp}
\mid l\lambda_k-\alpha_k-q_k\mid<\varepsilon,
\end{equation}
for each $k=1,2,\ldots,m$.
\end{theorem}

For brevity, whenever we have an inequality of the form $|\alpha-\beta|<\varepsilon$ for arbitrarily small $\varepsilon$, we will write instead $\alpha\approx\beta$ and omit the explicit dependence on $\varepsilon$. For example, (\ref{KroApp}) will be represented as $l\lambda_k-q_k\approx\alpha_k$.

When we study the PGST in vertex complemented coronas, the following result will be used to verify whether a set of numbers are linearly independent over the rational numbers.

\begin{theorem}
\emph{(see \cite[Theorem 1a]{Ri}) }
\label{Ri} Let $p_1,p_2,\ldots,p_k$ be distinct positive  primes.
Then the set $\left\{\sqrt[n]{p_1^{m(1)}\cdots p_k^{m(k)}}: 0\leq m(i)<n,~1\leq i \leq k \right\}$ is linearly independent over the set of rational numbers $\mathbb{Q}$.
\end{theorem}
When $n=2$, Theorem \ref{Ri} immediately implies the following result.
\begin{cor}
\label{independent}
The set $\left\{\sqrt{\Delta}: \Delta\text{~is~a~square-free~integer}\right\}$ is linearly independent over the set of rational numbers $\mathbb{Q}$.
\end{cor}

Let $G$ be a graph with adjacency matrix $A_G$.  The eigenvalues of $A_G$ are called the \emph{eigenvalues} of $G$. We use $\mathrm{Spec}_G$ to denote the set of all distinct eigenvalues of $G$. Suppose that $ \lambda_0>\lambda_1>\cdots>\lambda_p$ are all distinct eigenvalues of $G$ and $\left\{\mathbf{x}_{1}^{(j)}, \mathbf{x}_{2}^{(j)}, \ldots, \mathbf{x}_{r_j}^{(j)}\right\}$ is an orthonormal basis of the eigenspace associated with $\lambda_{j}$ with multiplicity $s_j$, $j=0,1,\ldots,p$.  Let $\mathbf{x}^H$ denote the conjugate transpose of a column vector $\mathbf{x}$. Then, for each eigenvalue $\lambda_j$ of $G$, define
$$
E_{\lambda_j} = \sum\limits_{i=1}^{r_j}\mathbf{x}_i^{(j)} \left(\mathbf{x}_i^{(j)}\right)^H,
$$
which is usually called the \emph{eigenprojector} (or orthogonal projector onto an eigenspace) corresponding to  $\lambda_j$ of $G$. Note that $\sum_{j=0}^pE_{\lambda_j}=I$ (the identity matrix). Then
\begin{equation}
\label{spect1}
A_G=A_G\sum_{j=0}^pE_{\lambda_j} =\sum_{j=0}^p\sum\limits_{i=1}^{r_j}A_G\mathbf{x}_i^{(j)} \left(\mathbf{x}_i^{(j)}\right)^H  =\sum_{j=0}^p\sum\limits_{i=1}^{r_j}\lambda_j\mathbf{x}_i^{(j)} \left(\mathbf{x}_i^{(j)}\right)^H  =\sum_{j=0}^{p}\lambda_jE_{\lambda_j},
\end{equation}
which is called the \emph{spectral decomposition of $A_G$ with respect to the distinct eigenvalues}  (see ``Spectral Theorem for Diagonalizable Matrices'' in \cite[Page 517]{MAALA}). Note that $E_{\lambda_j}^{2}=E_{\lambda_j}$ and $E_{\lambda_j}E_{\lambda_h}=\mathbf{0}$ for $j\neq h$, where $\mathbf{0}$ denotes the zero matrix. So, by (\ref{spect1}), we have
\begin{equation}\label{SpecDec2-1}
H_{A_G}(t)=\sum_{k\geq 0}\dfrac{(-\mathrm{i})^{k}A_G^{k}t^{k}}{k!}=\sum_{k\geq 0}\dfrac{(-\mathrm{i})^{k}\left(\sum\limits_{j=0}^{p}\lambda_{j}^{k}E_{\lambda_j}\right)t^{k}}{k!} =\sum_{j=0}^{p}\exp(-\mathrm{i}t\lambda_{j})E_{\lambda_j}.
\end{equation}

The \emph{eigenvalue support} of a vertex $u$ in $G$, denoted by $\mathrm{{supp}}_G(u)$, is the set of all eigenvalues $\lambda$ of $G$ such that $E_\lambda\mathbf{e}_u\neq \mathbf{0}$, where $\mathbf{e}_u$ is the characteristic vector corresponding to $u$. Two vertices $u$ and $v$ are \emph{strongly cospectral} if $E_\lambda\mathbf{e}_u=\pm E_\lambda\mathbf{e}_v$ for each eigenvalue $\lambda$ of $G$.

In the following, we state some useful results about PST and periodicity.

\begin{lemma}\label{periodic-pst}\emph{(see \cite[Lemma 2.1]{Godsil11})}
If $G$ has PST between vertices $u$ and $v$ at time $t$, then $G$ is periodic at $u$ at time $2t$.
\end{lemma}

\begin{lemma}\label{condition}\emph{(see \cite[Theorem 6.1]{Godsil12})}
A graph $G$ is periodic at vertex $u$ if and only if either:
\begin{itemize}
\item[\rm (a)]
all eigenvalues in $\mathrm{{supp}}_G(u)$ are integers; or
\item[\rm (b)]
there are square-free integer $\Delta$ and integer $a$ so that each eigenvalue $\lambda$ in $\mathrm{{supp}}_G(u)$ is of the form $\lambda=\frac{1}{2}\left(a+b_\lambda\sqrt{\Delta}\right)$, for some integer $b_\lambda$.
\end{itemize}

\end{lemma}

Coutinho gave a necessary and sufficient condition for a graph to have PST.

\begin{lemma}\label{pst}\emph{(see \cite[Theorem 2.4.4]{Coh14})}
Let $G$ be a graph and let $u,v$ be two distinct vertices of $G$. Then there exists PST between $u$ and $v$ at time $t$ if and only if all of the following conditions hold:
\begin{itemize}
\item[\rm (a)] Vertices $u$ and $v$ are strongly cospectral.
\item[\rm (b)] There are integers $a$ and $\Delta$, where $\Delta$ is square-free, so that for each eigenvalue $\lambda$ in $supp_G(u)$:
    \begin{itemize}
    \item[\rm (i)]$\lambda=\frac{1}{2}\left(a+b_\lambda\sqrt{\Delta}\right)$, for some integer $b_\lambda$.
    \item[\rm (ii)]$\mathbf{e}_u^\top E_\lambda(G)\mathbf{e}_v$ is positive if and only if $(\rho(G)-\lambda)/g\sqrt{\Delta}$ is even, where
        $$
        g:=\gcd\left(\left\{\frac{\rho(G)-\lambda}{\sqrt{\Delta}}:\lambda\in \mathrm{{supp}}_G(u)\right\}\right),
        $$
        and $\rho(G)$ denotes the largest eigenvalue of $G$.
    \end{itemize}
\end{itemize}
Moreover, if the above conditions hold, then there is a minimum time of PST between $u$ and $v$ given by $t_0:=\frac{\pi}{g\sqrt{\Delta}}$.
\end{lemma}

\section{Eigenvalues and eigenprojectors of vertex complemented coronas}

Before presenting the main results of this section, we first give some frequently used notations as follows.

\textbf{Notations.} Recall that $\mathbf{j}_m$ denotes the column vector of size $m$ with all entries equal to one, and let $J_{m\times n}$ denotes the $m\times n$ matrix with all entries equal to one. In particular, if $m=n$, we simply write $J_{m\times m}$ by $J_{m}$.
Let $\mathbf{e}_i^{n}$ denotes the unit vector of size $n$ with the $i$-th entry equal to $1$. If the size $n$ of $\mathbf{e}_i^{n}$ can be easily read from the context, then we can omit the superscript and write $\mathbf{e}_i^{n}$ as $\mathbf{e}_i$ for simplicity.  Let $\ast^\top$ denotes the transpose of $\ast$, where $\ast$ may be a vector or a matrix.

Let $G$ be a graph with vertex set $V(G)=\{v_1, v_2,\ldots,v_n\}$ and let $\overrightarrow{H}=(H_1,H_2,\ldots, H_n)$ be an $n$-tuple of graphs.
Formally, the vertex set of vertex complemented corona $G\tilde{\circ}\overrightarrow{H}$ can be labeled as follows:
\begin{equation*}
V(G\tilde{\circ}\overrightarrow{H}) =\left\{(v,0):v\in V(G)\right\}\cup\bigcup_{j=1}^{n}\left\{(v_j, w):v_j\in V(G), w\in V(H_j)\right\},
\end{equation*}
and the adjacency relation
\begin{equation*}
	\label{relation}
(v_i,w)\sim(v_j,w') \Longleftrightarrow \left\{
           \begin{array}{lr}
     w=w'=0\text{~and~}  v_i\sim  v_j \text{~in~} G, & \text{or}\\[0.2cm]
    v_i=v_j\text{~and~} w\sim w'  \text{~in~} {H_l},  &   \text{or}\\[0.2cm]
   v_i\neq v_j\text{~and~just~one~of~}w\text{~and~}w'\text{~is~}0 .& \end{array}\right.
\end{equation*}

If $G$ is a regular connected graph and $\overrightarrow{H}=(H_1,H_2,\ldots, H_n)$ is an $n$-tuple of regular graphs with $|V(H_i)|=m\geq1$ for $i=1,2,\ldots, n$, then we compute the eigenvalues of $G\tilde{\circ}\overrightarrow{H}$ in the following theorem.

\begin{theorem}\label{eigenprojector}
Let $G$ be an $r$-regular connected graph with $n\ge2$ vertices and let $\overrightarrow{H}=(H_1, H_2, \ldots, H_n)$ be an $n$-tuple of $k$-regular graphs with $|V(H_i)|=m\geq1$, $i=1,2,\ldots, n$. Suppose that $G$ has eigenvalues $r=\lambda_0>\lambda_1>\cdots>\lambda_p$ with multiplicities $1=s_0, s_1, \ldots, s_p$. Then the eigenvalues of $G\tilde{\circ}\overrightarrow{H}$ are
\begin{itemize}
    \item[\rm (a)]$k$ with multiplicity $\left(\sum\limits_{i=1}^n s^i_{k}\right)-n$, where $s^i_{k}$ denotes the multiplicity of  eigenvalue $k$ of $H_i$;
    \item[\rm (b)]$\mu$ with multiplicity $\sum\limits_{i=1}^n s^i_{\mu}$, where $\mu$ is an eigenvalue of $H_i$ with multiplicity $s^i_\mu$, which covers all eigenvalues of $H_i$ except for $\mu=k$, for $i=1,2,\ldots, n$;
    \item[\rm(c)]$\frac{1}{2}\left(\lambda_j+k\pm\sqrt{(\lambda_j-k)^2+4m}\right)$ with multiplicity $s_j$, for $j=1,2,\ldots, p$;
    \item[\rm(d)]$\frac{1}{2}\left(r+k\pm\sqrt{(r-k)^2+4m(n-1)^2}\right)$ with multiplicity $1$.
\end{itemize}

\end{theorem}

\begin{proof}
Define $M=J_n-I_n$. The adjacency matrix of $G\tilde\circ\overrightarrow{H}$ is given by
\begin{equation}\label{matrixGH}
	A_{G\tilde{\circ}\overrightarrow{H}}=
	\left(
	\begin{array}{cc}
	A_G& M\otimes \mathbf{j}_{m}^\top\\ [0.2cm]
	M^\top\otimes \mathbf{j}_{m}& \sum\limits_{i=1}^{n} \left(\mathbf{e}^n_i(\mathbf{e}^n_i)^\top\otimes A_{H_i}\right)
	\end{array}
	\right),
	\end{equation}
where $\otimes$ means the Kronecker product. By Lemma  \ref{suhur complemented}, the characteristic polynomial of $A_{G\tilde{\circ}\overrightarrow{H}}$ is
\begin{equation*}
\begin{aligned}
\det(xI_{n+nm}-A_{G\tilde\circ\overrightarrow{H}})&=\det\left(
	\begin{array}{cc}
	xI_{n}-A_G& -M\otimes \mathbf{j}_{m}^\top\\ [0.2cm]
	-M^\top\otimes \mathbf{j}_{m}& \sum\limits_{i=1}^{n} \left(\mathbf{e}^n_i(\mathbf{e}^n_i)^\top\otimes (xI_{m}-A_{H_i})\right)
	\end{array}
	\right)\\
&=\det(N)\det(S).
\end{aligned}
\end{equation*}
where
$$
N=\sum_{i=1}^{n}\left(\mathbf{e}^n_i(\mathbf{e}^n_i)^\top\otimes (xI_{m}-A_{H_i})\right),
$$
and
$$
S=xI_{n}-A_G-(M\otimes \mathbf{j}_{m}^\top)N^{-1}(M^\top\otimes \mathbf{j}_{m}).
$$
By Lemma \ref{gamma}, we have
$$
(M\otimes \mathbf{j}_{m}^\top)N^{-1}(M^\top\otimes \mathbf{j}_{m})=\frac{m}{x-k}MM^\top=\frac{m}{x-k}(I_n+(n-2)J_n).
$$
Then by Lemmas \ref{gamma} and \ref{liuzhang1}, we have
\begin{align*}
\det(S)&=\det\left(\left(x-\frac{m}{x-k}\right)I_n-A_G-\frac{m(n-2)}{x-k}J_n\right)\\
&=\left(1-\frac{m(n-2)}{x-k}\Gamma_{A_G}\left(x-\frac{m}{x-k}\right)\right)\cdot\det\left(\left(x-\frac{m}{x-k}\right)I_n-A_G\right)\\
&=(x-k)^{-n}\left(1-\frac{m(n-2)}{x-k}\cdot\frac{n}{x-\frac{m}{x-k}-r}\right)\cdot\det\left(\left(x(x-k)-m\right)I_n-(x-k)A_G\right)\\
&=(x-k)^{-n}\cdot\frac{(x-r)(x-k)-m-mn(n-2)}{(x-r)(x-k)-m}
\cdot\prod_{j=0}^{p}(x(x-k)-m-(x-k)\lambda_j)^{s_j}\\
&=(x-k)^{-n}\cdot\left((x-r)(x-k)-m(n-1)^2\right) \cdot\prod_{j=1}^{p}(x^2-(k+\lambda_j)x-m+k\lambda_j)^{s_j}.
\end{align*}
Note that
\begin{align*}
\det(N)=&\prod_{i=1}^{n} \det(xI_m-A _{H_i}).
\end{align*}
Therefore, the required result follows from  $\det(xI_{n+nm}-A_{G\tilde\circ\overrightarrow{H}})
=\det(N)\det(S)$.

This completes the proof.
\qed
\end{proof}

Next, by Theorem \ref{eigenprojector}, we compute the eigenprojectors of $G\tilde\circ\overrightarrow{H}$, where $G$ and $\overrightarrow{H}$ are as in Theorem \ref{eigenprojector}.

\begin{theorem}\label{theorem-eigenprojector}
Let $G$ and $\overrightarrow{H}$ be as in Theorem \ref{eigenprojector}. Then the eigenprojectors  of $G\tilde\circ\overrightarrow{H}$ are stated as follows:
\begin{itemize}
 \item[\rm (a)]
 $\mu$ is an eigenvalue of $G\tilde\circ\overrightarrow{H}$ with the eigenprojector
 \begin{align}\label{E1}
	E_{\mu}=
	\left(\begin{array}{cc}
	\mathbf{0}& \mathbf{0} \\ [0.1cm]
	\mathbf{0}& \sum\limits_{l=1}^{n} \left(\mathbf{e}^n_l(\mathbf{e}^n_l)^\top\right) \otimes \left(E_\mu(H_l)-\delta_{\mu, k}\cdot\frac{1}{m}J_{m}\right)
\end{array}
\right),
	\end{align}
where $E_\mu(H_l)$ denotes the eigenprojector corresponding to the eigenvalue $\mu$ of $H_l$ with the assumption that $E_\mu(H_l)=0$ if $\mu$ is not an eigenvalue of $H_l$, and $\delta_{\mu,k}$ is a function satisfying that
$$
\delta_{\mu,k}=\left\{ \begin{array}{ll}
                       1, & \mu=k, \\[0.2cm]
                       0, & \mu\neq k.
                     \end{array}
\right.
$$
Note that the case of $\mu=k$ occurs if and only if $H_l$ is disconnected.
\item[\rm (b)]
 For each eigenvalue $\lambda\neq r$ of $G$, $\lambda_\pm=\frac{1}{2}\left(\lambda+k\pm\sqrt{(\lambda-k)^2+4m}\right)$ are eigenvalues of $G\tilde\circ\overrightarrow{H}$ with the eigenprojectors
 \begin{align}\label{E2} E_{\lambda_\pm}=\frac{(\lambda_{\pm}-k)^2}{(\lambda_{\pm}-k)^2+m}
 \left(\begin{array}{cc}
               E_{\lambda}(G) & -\frac{1}{\lambda_{\pm}-k}E_{\lambda}(G) \otimes \mathbf{j}^\top_{m} \\
               -\frac{1}{\lambda_{\pm}-k} (E_{\lambda}(G))^\top\otimes \mathbf{j}_{m}& \frac{1}{(\lambda_{\pm}-k)^2}E_{\lambda}(G) \otimes J_{m}
             \end{array}
\right),
	\end{align}
where $M=J_n-I_n$ and $E_{\lambda}(G)$ denotes the eigenprojector corresponding to eigenvalue $\lambda$ of $G$.

\item[\rm (c)]$r_\pm=\frac{1}{2}\left(r+k\pm\sqrt{(r-k)^2+4m(n-1)^2}\right)$ are eigenvalues of $G\tilde\circ\overrightarrow{H}$ with the eigenprojectors
\begin{align}\label{E3}
E_{r_\pm}=\frac{(r_{\pm}-k)^2}{(r_{\pm}-k)^2+m(n-1)^2}
\left(\begin{array}{cc}
               E_r(G) & \frac{n-1}{r_{\pm}-k} E_r(G)\otimes \mathbf{j}^\top_{m} \\[0.25cm]
               \frac{n-1}{r_{\pm}-k} (E_r(G))^\top\otimes \mathbf{j}_{m}& \frac{(n-1)^2}{(r_{\pm}-k)^2} E_r(G)\otimes J_m
             \end{array}
\right).
	\end{align}
\end{itemize}

Therefore, the spectral decomposition of $A_{G\tilde\circ\overrightarrow{H}}$ is given by
\begin{equation}\label{Spec}
	A_{G\tilde\circ\overrightarrow{H}} = \left(\sum\limits_{\lambda\in \mathrm{Spec}_G }\sum\limits_{\pm}\lambda_\pm E_{\lambda_\pm}\right)+\sum\limits_\mu\mu E_\mu,
	\end{equation}
where $\mu$ covers all eigenvalues of $H_l$, $l=1,2,\ldots, n$.
\end{theorem}

\begin{proof}
The proofs of (a)--(c) consist of Claims 1--3.

\smallskip

\begin{claim}
{\em
$\mathbf{X}$, $\mathbf{Y}_\pm$ and $\mathbf{Z}_\pm$ defined below
 are eigenvectors of $A_{G\tilde\circ\overrightarrow{H}}$ corresponding to eigenvalues $\mu$, $\lambda_\pm$ and $r_\pm$, respectively.
 }
 \end{claim}

\noindent\textbf{Proof of Claim 1.}
Let $H_l$ be a graph in $\overrightarrow{H}$. Note that $k$ is always an eigenvalue of $H_l$ with an eigenvector $\mathbf{j}_m$. Note also that $E_k(H_l)=\frac{1}{m}J_m$ if and only if $H_l$ is connected. Suppose that $\mathbf{x}\bot\mathbf{ j}_m$ is an eigenvector of $A_{H_l}$ corresponding to the eigenvalue $\mu$ of $H_l$ (Here, $\mu$ may be equal to $k$, and the case of $\mu=k$ occurs if and only if $H_l$ is disconnected). Define
$$
\mathbf{X}:=\left(\begin{array}{cc}
\mathbf{0}_{n\times1}\\ [0.2cm]
\mathbf{e}_l^n\otimes \mathbf{x}
\end{array}\right),
$$
where $\mathbf{0}_{s\times t}$ denotes the $s\times t$ matrix with all entries equal to $0$. Notice that the adjacency matrix $A_{G\tilde\circ\overrightarrow{H}}$ is given in (\ref{matrixGH}). Then, we have
\begin{equation}\label{mueigen}
A_{G\tilde\circ\overrightarrow{H}}\mathbf{X}=\mu\mathbf{X}.
\end{equation}
Thus, $\mathbf{X}$ is an eigenvector of $A_{G\tilde\circ\overrightarrow{H}}$ with the eigenvalue $\mu$.

Suppose that $\mathbf{y}\perp \mathbf{j}_{n}$ is a unit eigenvector of $A_{G}$  corresponding to the eigenvalue $\lambda\neq r$.
Define
$$
\mathbf{Y}_\pm:=\left(\begin{array}{cc}
 \mathbf{y}\\[0.2cm]
 -\frac{1}{\lambda_\pm-k}\mathbf{y}\otimes \mathbf{j}_m
 \end{array}
 \right).
 $$
 Note that $M\mathbf{y}=M^\top\mathbf{y}=-\mathbf{y}$, and keep in mind that $\mathbf{y}$ can be regarded as $\mathbf{y}\otimes 1$. Then if $\lambda\neq r$, by (\ref{matrixGH}), we have
\begin{align*}
 A_{G\tilde\circ\overrightarrow{H}}\mathbf{Y}_\pm
&=\left(
\begin{array}{cc}
A_G& M\otimes \mathbf{j}_{m}^\top\\ [0.2cm]
	M^\top\otimes \mathbf{j}_{m}& \sum\limits_{i=1}^{n} \left(\mathbf{e}^n_i(\mathbf{e}^n_i)^\top\otimes A_{H_i}\right)
\end{array}\right)
\left(\begin{array}{c}
 \mathbf{y}\\[0.25cm]
 -\frac{1}{\lambda_\pm-k}\mathbf{y}\otimes \mathbf{j}_m
\end{array}\right)\\
&=\left(\begin{array}{c}
\lambda\mathbf{y}-\frac{1}{\lambda_\pm-k}(M\mathbf{y}) \otimes\mathbf{j}_m^\top\mathbf{j}_m\\[0.25cm]
(M^\top\mathbf{y})\otimes\mathbf{j}_m- \frac{k}{\lambda_\pm-k}\mathbf{y}\otimes\mathbf{j}_m
\end{array}\right)\\
&=\left(\begin{array}{c}\lambda \mathbf{y}+\frac{m}{\lambda_\pm-k}\mathbf{y}\\[0.25cm]
-\frac{\lambda_\pm}{\lambda_\pm-k}\mathbf{y}\otimes\mathbf{j}_m
\end{array}\right)\\
&=\lambda_\pm\mathbf{Y}_\pm.
\end{align*}
Thus, $\mathbf{Y}_\pm$ are eigenvectors of $A_{G\tilde\circ\overrightarrow{H}}$ with eigenvalues $\lambda_\pm$.

Let $\mathbf{z}=\frac{1}{\sqrt{n}}\mathbf{j}_n$. Note that $A_{G}\mathbf{z}=r\mathbf{z}$. Define
$$
\mathbf{Z}_\pm:=\left(\begin{array}{cc}
 \mathbf{z}\\[0.2cm]
 \frac{n-1}{r_\pm-k}\mathbf{z}\otimes \mathbf{j}_m
 \end{array}
 \right).
 $$
 Note that $M\mathbf{z}=M^\top\mathbf{z}=(n-1)\mathbf{z}$, and keeping in mind that $\mathbf{z}$ can be regarded as $\mathbf{z}\otimes 1$, by (\ref{matrixGH}), we have
\begin{align*}
A_{G\tilde\circ\overrightarrow{H}}\mathbf{Z}_\pm
&=\left(\begin{array}{c}
r\mathbf{z}+\frac{n-1}{r_\pm-k}(M\mathbf{z}) \otimes\mathbf{j}_m^\top\mathbf{j}_m\\[0.25cm]
(M^\top\mathbf{z})\otimes\mathbf{j}_m+ \frac{k(n-1)}{r_\pm-k}\mathbf{z}\otimes\mathbf{j}_m
\end{array}\right)\\
&=\left(\begin{array}{c}
r \mathbf{z}+\frac{m(n-1)^2}{r_\pm-k}\mathbf{z}\\[0.25cm]
\frac{r_\pm(n-1)}{r_\pm-k}\mathbf{z}\otimes\mathbf{j}_m
\end{array}\right)\\
&=r_\pm\mathbf{Z}_\pm.
\end{align*}
Thus, $\mathbf{Z}_\pm$ are eigenvectors of $A_{G\tilde\circ\overrightarrow{H}}$ with eigenvalues $r_\pm$.

\begin{claim}
{\em All $\mathbf{X}$'s, $\mathbf{Y}_\pm$'s and $\mathbf{Z}_\pm$'s are orthogonal eigenvectors of  $A_{G\tilde\circ\overrightarrow{H}}$.}
\end{claim}

\noindent\textbf{Proof of Claim 2.}
Recall that $\mathbf{x}\bot\mathbf{j}_m$, $\mathbf{y}\bot\mathbf{j}_n$ and $\mathbf{z}=\frac{1}{\sqrt{n}}\mathbf{j}_n$. Then one can easily verify that $\mathbf{X}\bot\mathbf{Y}_\pm$, $\mathbf{X}\bot\mathbf{Z}_\pm$ and $\mathbf{Y}_\pm \bot\mathbf{Z}_\pm$.

Consider $\mathbf{X}=\left(\begin{array}{cc}
\mathbf{0}_{n\times1}\\ [0.2cm]
\mathbf{e}_l^n\otimes \mathbf{x}
\end{array}\right)$
 and $\mathbf{X}'=\left(\begin{array}{cc}
\mathbf{0}_{n\times1}\\ [0.2cm]
\mathbf{e}_l^n\otimes \mathbf{x}'
\end{array}\right)$, where $\mathbf{x}$ and $\mathbf{x}'$ are orthogonal eigenvectors in $H_l$. Clearly, $\mathbf{X}\bot\mathbf{X}'$.

Consider $\mathbf{Y}_\pm=\left(\begin{array}{cc}
 \mathbf{y}\\[0.25cm]
 -\frac{1}{\lambda_\pm-k} \mathbf{y}\otimes \mathbf{j}_m
 \end{array}
 \right)$
 and $\mathbf{Y}'_\pm=\left(\begin{array}{cc}
 \mathbf{y}'\\[0.25cm]
 -\frac{1}{\lambda'_\pm-k} \mathbf{y}'\otimes \mathbf{j}_m
 \end{array}
 \right)$, where $\mathbf{y}$ and $\mathbf{y}'$ are unit orthogonal eigenvectors of $A_G$ corresponding to $\lambda\not=r$ and $\lambda'\not=r$ (Here, $\lambda$ and $\lambda'$ may be equal). Note that $\mathbf{y}\perp \mathbf{j}_{n}$, $\mathbf{y}'\perp \mathbf{j}_{n}$ and $\mathbf{y}\perp \mathbf{y}'$. Thus,
\begin{align*}
(\mathbf{Y}_{\pm})^\top\mathbf{Y}'_{\pm} &=\mathbf{y}^\top\mathbf{y}'+
\frac{\mathbf{y}^\top\mathbf{y}'\otimes \mathbf{j}^\top_{m}\mathbf{j}_{m}}{(\lambda_\pm-k)(\lambda'_\pm-k)}=0,
\end{align*}
that is, $\mathbf{Y}_\pm\bot\mathbf{Y}'_\pm$.

Consider $\mathbf{Y}_+$ and $\mathbf{Y}_-$. Recall that $\mathbf{y}\bot\mathbf{j}_n$. Note that
\begin{align*}
(\lambda_+-k)(\lambda_--k)=-m.
\end{align*}
Thus,
\begin{align*}
\mathbf{Y}_+\mathbf{Y}_-=\mathbf{y}^\top\mathbf{y}+
\frac{\mathbf{y}^\top\mathbf{y}\otimes \mathbf{j}^\top_{m}\mathbf{j}_{m}}{(\lambda_+-k)(\lambda_--k)}=1-\frac{m}{m}=0,
\end{align*}
that is, $\mathbf{Y}_+\bot\mathbf{Y}_-$.

Consider $\mathbf{Z}_+$ and $\mathbf{Z}_-$. Recall that $\mathbf{z}=\frac{1}{\sqrt{n}}\mathbf{j}_n$. Note that
\begin{align*}
(r_+-k)(r_--k)=-m(n-1)^2.
\end{align*}
Thus,
\begin{align*}
\mathbf{Z}_+\mathbf{Z}_-=\mathbf{z}^\top\mathbf{z}+
\frac{(n-1)^2\mathbf{z}^\top\mathbf{z}\otimes \mathbf{j}^\top_{m}\mathbf{j}_{m}}{(r_+-k)(r_--k)} =1-\frac{m(n-1)^2}{m(n-1)^2} =0,
\end{align*}
that is, $\mathbf{Z}_+\bot\mathbf{Z}_-$.

\begin{claim}
{\em
(\ref{E1}), (\ref{E2}) and (\ref{E3}) are eigenprojectors of $A_{G\tilde\circ\overrightarrow{H}}$ corresponding to eigenvalues $\mu$, $\lambda_\pm$ and $r_\pm$, respectively.
}
\end{claim}

\noindent\textbf{Proof of Claim 3.}
By (\ref{mueigen}), one can easily verify that (\ref{E1}) is the eigenprojector corresponding to the eigenvalue $\mu$.

Suppose that $\left\{\mathbf{y}^{(1)}, \mathbf{y}^{(2)}, \ldots,\mathbf{y}^{(s)}\right\}$ is a unit orthonormal basis of the eigenspace of $G$ corresponding to the eigenvalue $\lambda\not=r$. Set
\begin{align*}
\mathbf{Y}^{(i)}_{\pm}:=\left(\begin{array}{c}
                                                                 \mathbf{y}^{(i)}\\[0.25cm]
                                                                 -\frac{1}{\lambda_{\pm}-k}\mathbf{y}^{(i)} \otimes \mathbf{j}_{m}
                                                               \end{array}
\right).
\end{align*}
Then
\begin{equation*}
   \left\|\mathbf{Y}^{(i)}_{\pm}\right\|^2
 =1+\frac{m}{(\lambda_{\pm}-k)^2}.
 \end{equation*}
Let $E_{\lambda}(G)= \sum\limits_{i=1}^s\left(\mathbf{y}^{(i)}\right)\left(\mathbf{y}^{(i)}\right)^\top$ be the eigenprojector of $G$ corresponding to the eigenvalue $\lambda\not=r$. Then eigenprojectors of $A_{G\tilde\circ\overrightarrow{H}}$ corresponding to $\lambda_{\pm}$ are given as follows:
\begin{align*}
&E_{\lambda_{\pm}}(G\tilde\circ\overrightarrow{H}) \\ =&\frac{(\lambda_{\pm}-k)^2}{(\lambda_{\pm}-k)^2+m} \cdot \sum_{i=1}^{s} \mathbf{Y}^{(i)}_{\pm} \left(\mathbf{Y}^{(i)}_{\pm}\right)^\top\\
=&\frac{(\lambda_{\pm}-k)^2}{(\lambda_{\pm}-k)^2+m}
 \left(\begin{array}{cc}
               E_{\lambda}(G) & -\frac{1}{\lambda_{\pm}-k}E_{\lambda}(G) \otimes \mathbf{j}^\top_{m} \\
               -\frac{1}{\lambda_{\pm}-k} (E_{\lambda}(G))^\top\otimes \mathbf{j}_{m}& \frac{1}{(\lambda_{\pm}-k)^2}E_{\lambda}(G) \otimes J_{m}
             \end{array}
\right) ,
\end{align*}
yielding (\ref{E2}).

Since
\begin{equation*}
   \left\|\mathbf{Z}_{\pm}\right\|^2
 =1+\frac{m(n-1)^2}{(r_{\pm}-k)^2}.
 \end{equation*}
 Then eigenprojectors  of $A_{G\tilde\circ\overrightarrow{H}}$ corresponding to $r_{\pm}$  are given as follows:
\begin{align*}
&E_{r_{\pm}}(G\tilde\circ\overrightarrow{H}) \\ =&\frac{(r_{\pm}-k)^2}{(r_{\pm}-k)^2+m(n-1)^2} \mathbf{Z}_{\pm} \left(\mathbf{Z}_{\pm}\right)^\top\\
=&\frac{(r_{\pm}-k)^2}{(r_{\pm}-k)^2+m(n-1)^2}
\left(\begin{array}{cc}
               E_r(G) & \frac{n-1}{r_{\pm}-k} E_r(G)\otimes \mathbf{j}^\top_{m} \\[0.25cm]
               \frac{n-1}{r_{\pm}-k} (E_r(G))^\top\otimes \mathbf{j}_{m}& \frac{(n-1)^2}{(r_{\pm}-k)^2} E_r(G)\otimes J_m
             \end{array}
\right),
\end{align*}
yielding (\ref{E3}).

At last, it is easy to verify that (\ref{Spec}) is  the spectral decomposition of $A_{G\tilde\circ\overrightarrow{H}}$.

This completes the proof.
\qed\end{proof}

\section{State transfers in vertex complemented coronas}
\subsection{PST in vertex complemented coronas}

In this section, we prove that PST in vertex complemented coronas is extremely rare. In order to prove such a result, Lemma \ref{periodic-pst} implies that we just need to verify there is no periodic vertex in vertex complemented coronas.

\begin{lemma}\label{vw-v0}
Let $G$ and $\overrightarrow{H}$ be as in Theorem \ref{eigenprojector}. If $(v,w)$ is a periodic vertex of $G\tilde\circ\overrightarrow{H}$, then $(v,0)$ a periodic vertex of $G\tilde\circ\overrightarrow{H}$.
\end{lemma}

\begin{proof}
By Theorem \ref{theorem-eigenprojector}, the eigenvalue support of $(v,0)$ is contained in the eigenvalue support of $(v,w)$.
\qed\end{proof}

Next we show a necessary and sufficient condition for periodicity in vertex complemented coronas.

\begin{lemma}\label{mainlemma}
Let $G$ and $\overrightarrow{H}$ be as in Theorem \ref{eigenprojector}, and let $\lambda_\pm$ with $\lambda\not=r$ and $r_\pm$ be as in Theorem \ref{theorem-eigenprojector}.
\begin{itemize}
 \item[\rm (a)] If $r\not=k$, then $(v,0)$ is a periodic vertex of $G\tilde\circ\overrightarrow{H}$ if and only if  for each eigenvalue $\lambda\in \mathrm{{supp}}_G(v)\setminus\{r\}$, all $\lambda-k$, $\sqrt{(\lambda-k)^2+4m}$ and $\sqrt{(r-k)^2+4m(n-1)^2}$ are integers.

\item[\rm (b)] If $r=k$, then $(v,0)$ is a periodic vertex of $G\tilde\circ\overrightarrow{H}$ if and only if  there exists a positive square-free integer $\Delta$ such that for each eigenvalue $\lambda\in \mathrm{{supp}}_G(v)\setminus\{r\}$, all $\lambda-k$, $\sqrt{(\lambda-k)^2+4m}$ and $\sqrt{4m(n-1)^2}$ are integer multiples of $\sqrt{\Delta}$. Moreover, if this holds, then $\Delta\mid m$.
\end{itemize}
\end{lemma}

\begin{proof}
By Theorem \ref{theorem-eigenprojector}, the eigenvalue support of $(v,0)$ is given by $\mathrm{{supp}}_{G\tilde\circ\overrightarrow{H}}((v,0))=\{\lambda_\pm:\lambda\in \mathrm{{supp}}_G(v)\}$. Moreover, $r_\pm$ are always in $\mathrm{{supp}}_{G\tilde\circ\overrightarrow{H}}((v,0))$.

(a) For the sufficiency, for each eigenvalue $\lambda\in \mathrm{{supp}}_G(v)\setminus\{r\}$, all $\lambda-k$, $\sqrt{(\lambda-k)^2+4m}$ and $\sqrt{(r-k)^2+4m(n-1)^2}$ are integers. Clearly, $\lambda_\pm\in\mathrm{{supp}}_{G\tilde\circ\overrightarrow{H}}((v,0))$ and $r_\pm$ are integers. By Lemma \ref{condition}, $(v,0)$ is a periodic vertex.

For the necessity,  by Lemma \ref{condition}, we consider the following two cases.

\emph{Case 1.} All eigenvalues in $\mathrm{{supp}}_{G\tilde\circ\overrightarrow{H}}((v,0))$ are integers. In this case, $\lambda-k=\lambda_++\lambda_--2k~(\lambda\neq r)$, $\sqrt{(\lambda-k)^2+4m}=\lambda_+-\lambda_-~(\lambda\neq r)$ and $\sqrt{(r-k)^2+4m(n-1)^2}=r_+-r_-$ are integers.

\emph{Case 2.} There are integer $a$ and square-free integer $\Delta\ge2$ such that each eigenvalue $\lambda_\pm\in \mathrm{{supp}}_{G\tilde\circ\overrightarrow{H}}((v,0))$ is of the form $\lambda_\pm=\frac{1}{2}(a+b_{\lambda_\pm}\sqrt{\Delta})$, where $b_{\lambda_\pm}$ are integers corresponding to eigenvalues $\lambda_\pm$.
Recall that $(\lambda_+-k)(\lambda_--k)=-m$ for $\lambda\neq r$ and $(r_+-k)(r_--k)=-m(n-1)^2$. Then, in this case, we have
$$
-m=\frac{1}{4}\left((a-2k)^2+b_{\lambda_+}b_{\lambda_-}\Delta\right) +\frac{1}{4}(a-2k)(b_{\lambda_+}+b_{\lambda_-})\sqrt{\Delta},
$$
and
$$
-m(n-1)^2=\frac{1}{4}\left((a-2k)^2+b_{r_+}b_{r_-}\Delta\right) +\frac{1}{4}(a-2k)(b_{r_+}+b_{r_-})\sqrt{\Delta}.
$$
Note that $\sqrt{\Delta}$ is irrational. Then we have $a-2k=0$ or $b_{\lambda_+}+b_{\lambda_-}=0$ for each $\lambda\in\mathrm{{supp}}_G(v)$.

\emph{Case 2.1. } $b_{\lambda_+}+b_{\lambda_-}=0$ for each $\lambda\in\mathrm{{supp}}_G(v) $. In this case, we have $a=\lambda_++\lambda_-=\lambda+k$ and $a=r_++r_-=r+k$. Thus, $\mathrm{{supp}}_G(v)=\{r\}$, that is, $|\mathrm{{supp}}_G(v)|=1$. This is a contradiction to that $G$ is a connected graph with $n\geq2$ vertices.

\emph{Case 2.2. } $a-2k=0$. This implies that $\lambda_\pm=k+\frac{1}{2}b_{\lambda_\pm}\sqrt{\Delta}$. Hence, for $\lambda=r$,
\begin{align*}
\frac{1}{2}(b_{r_+}+b_{r_-})\sqrt{\Delta} &=(r_+-k)+(r_--k)=r-k, 
\end{align*}
Clearly, one side of the above equation is integer and the other side is irrational, this is a contradiction.

(b) For the sufficiency, if there exists a positive square-free integer $\Delta$ such that for each eigenvalue $\lambda\in \mathrm{{supp}}_G(v)\setminus\{r\}$, all of the following conditions hold:
\begin{equation*}
\lambda-k=e_\lambda\sqrt{\Delta}, ~\sqrt{(\lambda-k)^2+4m}=f_\lambda\sqrt{\Delta} ~\text{and}~ \sqrt{4m(n-1)^2}=f_r\sqrt{\Delta},
\end{equation*}
where $e_\lambda$ and $f_\lambda$ are integers corresponding to $\lambda$, then
$$\lambda_\pm=\frac{1}{2}\left(2k+(e_\lambda\pm f_\lambda)\sqrt{\Delta}\right)~ (\lambda\neq r) \text{~and~} r_\pm=\frac{1}{2}\left(2k\pm f_r\sqrt{\Delta}\right).$$
By Lemma \ref{condition}, $(v,0)$ is a periodic vertex of $G\tilde\circ\overrightarrow{H}$.

For the necessity, by Lemma \ref{condition}, we consider the following two cases.

\emph{Case 1.} All eigenvalues in $\mathrm{{supp}}_{G\tilde\circ\overrightarrow{H}}((v,0))$ are integers. In this case, $\lambda-k=\lambda_++\lambda_--2k~(\lambda\neq r)$, $\sqrt{(\lambda-k)^2+4m}=\lambda_+-\lambda_-~(\lambda\neq r)$ and $\sqrt{4m(n-1)^2}=r_+-r_-$ are integers.

\emph{Case 2.} There are integer $a$ and square-free integer $\Delta\ge2$ such that each eigenvalue $\lambda_\pm\in \mathrm{{supp}}_{G\tilde\circ\overrightarrow{H}}((v,0))$ is of the form $\lambda_\pm=\frac{1}{2}(a+b_{\lambda_\pm}\sqrt{\Delta})$, where $b_{\lambda_\pm}$ are integers corresponding to eigenvalues $\lambda_\pm$.
Similar to the proof of Case 2 of (a), we have $a-2k=0$ or $b_{\lambda_+}+b_{\lambda_-}=0$ for each $\lambda\in\mathrm{{supp}}_G(v)$.  If $b_{\lambda_+}+b_{\lambda_-}=0$ for each $\lambda\in\mathrm{{supp}}_G(v) $, similar to the proof of Case 2.1 of (a), we also obtain a contradiction to that $G$ is a connected graph with $n\geq2$ vertices. If $a-2k=0$, then we have $\lambda_\pm=k+\frac{1}{2}b_{\lambda_\pm}\sqrt{\Delta}$. Hence,
\begin{align*}
\frac{1}{2}(b_{\lambda_+}+b_{\lambda_-})\sqrt{\Delta} &=(\lambda_+-k)+(\lambda_--k)=\lambda-k\text{~~for~}\lambda\neq r,\\[0.2cm]
\frac{1}{2}(b_{\lambda_+}-b_{\lambda_-})\sqrt{\Delta} & =\lambda_+-\lambda_- =\sqrt{(\lambda-k)^2+4m}\text{~~for~}\lambda\neq r,
\end{align*}
and
\begin{align*}
\frac{1}{2}(b_{r_+}-b_{r_-})\sqrt{\Delta}
&=r_+-r_-=\sqrt{4m(n-1)^2}.
\end{align*}
The above three equations imply that $\lambda-k$ for $\lambda\neq r$, $\sqrt{(\lambda-k)^2+4m}$ for $\lambda\neq r$ and $\sqrt{4m(n-1)^2}$ are of the form $x\sqrt{\Delta}/2$, where $x\in \mathbb{Z}$. Note that their squares are rational algebraic integers. Thus, their squares must be integers. Therefore,  $\lambda-k$ for $\lambda\neq r$, $\sqrt{(\lambda-k)^2+4m}$ for $\lambda\neq r$ and $\sqrt{4m(n-1)^2}$ are integer multiples of $\sqrt{\Delta}$.

The condition $\sqrt{4m(n-1)^2}$ is an integer multiple  of $\sqrt{\Delta}$ implies that $\Delta\mid m$ immediately.
\qed\end{proof}

By Lemma \ref{mainlemma}, we have the following result.

\begin{cor}\label{Kmlemma}
Let $G$ and $\overrightarrow{H}$ be as in Theorem \ref{eigenprojector}. If $(v,0)$ is a periodic vertex of $G\tilde\circ\overrightarrow{H}$, then
$$
m\geq|\lambda-k|+1 \mathrm{~for~}\lambda\in \mathrm{{supp}}_G(v)\setminus\{r\},
$$
and
$$
m(n-1)^2\geq |r-k|+1.
$$
\end{cor}

\begin{proof}
\emph{Case 1.} $r\neq k$.
If $(v,0)$ is a periodic vertex of $G\tilde\circ\overrightarrow{H}$, then by Lemma \ref{mainlemma} (a), for each eigenvalue $\lambda\in \mathrm{{supp}}_G(v)\setminus\{r\}$, all $(\lambda-k)^2$, $(r-k)^2$, $(\lambda-k)^2+4m$ and $(r-k)^2+4m(n-1)^2$ are squares. Since $4m$ and $4m(n-1)^2$ are even, $(\lambda-k)^2$ and$(\lambda-k)^2+4m$ have the same parity. Similarly, $(r-k)^2$ and $(r-k)^2+4m(n-1)^2$ have the same parity. Hence,
\begin{equation*}
4m\ge \left(|\lambda-k|+2\right)^2- |\lambda-k|^2 =4\left(|\lambda-k|+1\right) \text{~for~} \lambda\in \mathrm{{supp}}_G(v)\setminus\{r\},
\end{equation*}
and
\begin{equation*}
4m(n-1)^2 \ge \left(| r-k|+2\right)^2-\left(| r-k|\right)^2 = 4\left(|r-k|+1\right).
\end{equation*}
The required result is obtained by simplifying the above inequalities immediately.

\emph{Case 2.} $r=k$.
If $(v,0)$ is a periodic vertex of $G\tilde\circ\overrightarrow{H}$, then by Lemma \ref{mainlemma} (b), there exists a positive square-free integer $\Delta$ such that for each eigenvalue $\lambda\in \mathrm{{supp}}_G(v)\setminus\{r\}$, both $(\lambda-k)^2/\Delta$ and $\left((\lambda-k)^2+4m\right)/\Delta$ are squares. Recall that $(r-k)^2$ and $(\lambda-k)^2+4m$ have the same parity and $\Delta\mid m$. Hence,
\begin{equation*}
\frac{4m}{\Delta}\ge\left(\frac{|\lambda-k|}{\sqrt{\Delta}}+2\right)^2 -\left(\frac{|\lambda-k|}{\sqrt{\Delta}}\right)^2= 4\left(\frac{|\lambda-k|}{\sqrt{\Delta}}+1\right) \text{~for~} \lambda\in \mathrm{{supp}}_G(v)\setminus\{r\}.
\end{equation*}
Since $\Delta\geq1$, we have
\begin{equation*}
m\geq|\lambda-k|\sqrt{\Delta}+\Delta\geq|\lambda-k|+1\text{~for~} \lambda\in \mathrm{{supp}}_G(v)\setminus\{r\}.
\end{equation*}
Furthermore, $m(n-1)^2\geq1$ and then the second inequality holds.

This completes the proof.
\qed \end{proof}

As an application of Corollary \ref{Kmlemma}, we prove that there is no PST in vertex complemented corona $G\tilde\circ \overrightarrow{H}$, where $G$ is an $r$-regular connected graph, $\overrightarrow{H}=(K_m, K_m, \ldots, K_m)$ and $K_m$ denotes a complete graph on $m$ vertices. For the sake of simplicity, such a graph will be denoted by $G\tilde\circ K_m$.

\begin{cor}\label{kmpst}
Let $G$ be as in Theorem \ref{eigenprojector}. Then every vertex of $G\tilde\circ K_m$ is not periodic. Moreover, $G\tilde\circ K_m$ has no PST.
\end{cor}

\begin{proof}
Suppose that the vertex $(v,0)$ is a periodic vertex of $G\tilde\circ K_m$, where $v$ is a vertex of $G$. We claim that there exists a negative eigenvalue in the eigenvalue support of $v$ in $G$. Otherwise, assume that every eigenvalue in $\mathrm{{supp}}_G(v)$ is non-negative. Then $E_\lambda(G) \mathbf{e}_v=\mathbf{0}$ for each negative eigenvalue $\lambda\in \mathrm{{Spec}}_G$.
Note that
$$
\mathbf{e}_v^{\top}  A_G \mathbf{e}_v =\sum_{\lambda\in \mathrm{{Spec}}_G}\lambda \mathbf{e}_v^{\top} E_\lambda(G) \mathbf{e}_v=0.
$$
Then $\mathbf{e}_v^{\top} E_\lambda(G) \mathbf{e}_v=0$ for each positive eigenvalue $\lambda\in \mathrm{{Spec}}_G$. Note that $E_r(G)=\frac{1}{n}J_n$ and thus $\mathbf{e}_v^{\top} E_r(G) \mathbf{e}_v=\frac{1}{n}\ne0$, a contradiction.  Hence, there exists a negative eigenvalue $\lambda<0$ in $\mathrm{{supp}}_G(v)$. Then, $\lambda-(m-1)<0$. By Corollary \ref{Kmlemma}, we have
$$
m\geq\mid\lambda-(m-1)\mid+1=-\lambda+(m-1)+1>m,
$$
a contradiction. Therefore, $(v,0)$ is not a periodic vertex of $G\tilde\circ K_m$. By Lemma \ref{vw-v0}, we conclude that every vertex of $G\tilde\circ K_m$ is not periodic. Moreover, by Lemma \ref{periodic-pst}, $G\tilde\circ K_m$ has no PST.
\qed \end{proof}

By Lemma \ref{periodic-pst}, we know that periodicity is a necessary condition for a graph to have PST. In the following, we give a sufficient condition for a vertex  complemented corona to not be periodic.

\begin{theorem}\label{miantheorem}
Let $G$ and $\overrightarrow{H}$ be as in Theorem \ref{eigenprojector}, and let $v$ be a vertex of $G$.
\begin{itemize}
\item[\rm (a)] If there are two distinct eigenvalues $\lambda,\mu\in \mathrm{{supp}}_G(v)\setminus\{r\}$ such that
    \begin{equation}\label{mainE1}
  |\lambda-k|-|\mu-k|\in\left\{\sqrt{\Delta}, 2\sqrt{\Delta}\right\}
    \end{equation}
  for some square-free integer $\Delta$, then $(v,w)$ is not a periodic vertex of $G\tilde\circ\overrightarrow{H}$, for all $w\in V(H_i)\cup\{0\}$.
\item[\rm (b)] If there is an eigenvalue $\kappa\in \mathrm{{supp}}_G(v)\setminus\{r\}$ such that
\begin{equation}\label{mainE2}
\big||r-k|-(n-1)|\kappa-k |\big|\in\left\{\sqrt{\Delta},2\sqrt{\Delta}\right\}
\end{equation}
for some square-free integer $\Delta$, then $(v,w)$ is not a periodic vertex of $G\tilde\circ\overrightarrow{H}$, for all $w\in V(H_i)\cup\{0\}$.
\end{itemize}
\end{theorem}

\begin{proof}
(a) By Lemma \ref{vw-v0}, we just need to show that $(v,0)$ is not a periodic vertex of $G\tilde\circ\overrightarrow{H}$. By contradiction, suppose that $(v,0)$ is a periodic vertex. By Lemma \ref{mainlemma}, there exists a square-free integer $\Delta\geq1$ such that for each eigenvalue $\lambda\in \mathrm{{supp}}_G(v)\setminus\{r\}$, both $\lambda-k$ and $\sqrt{(\lambda-k)^2+4m}$ are integer multiples of $\sqrt{\Delta}$. Define
$$
\delta:=\frac{1}{\sqrt{\Delta}}\min\left\{\big||\lambda_1-k|-|\lambda_2-k|\big|: \lambda_1, \lambda_2\in \mathrm{{supp}}_G(v)\setminus\{r\}\right\}.
$$
Assume that $\lambda$ and $\mu$ are two eigenvalues achieving the above minimum. Define
$$
n_\lambda:=\frac{|\lambda-k|}{\sqrt{\Delta}}, \text{~and~} n_\mu:=\frac{|\mu-k|}{\sqrt{\Delta}},
$$
and suppose that $\delta=n_\lambda-n_\mu$.
It is already noted in the beginning of the proof that $n_\lambda^2+4m/\Delta$ and $n_\mu^2+4m/\Delta$ are squares. Define
\begin{equation*}
p:=\sqrt{n_\mu^2+\frac{4m}{\Delta}}, \text{~and~} q:=\sqrt{n_\lambda^2+\frac{4m}{\Delta}}.
\end{equation*}
Then
\begin{equation*}
q+p>n_\lambda+n_\mu=2n_\mu+\delta,  \text{~and~} q^2-p^2=(2n_\mu+\delta)\delta,
\end{equation*}

\noindent which implies $q-p<\delta$. By (\ref{mainE1}), we have $\delta=1,2$. If $\delta=1$, then $q-p<1$, which cannot occur.
If $\delta=2$, then $q-p<2$, which contradicts that $p$ and $q$ have the same parity.

(b) Similar to the proof of (a), suppose that $(v,0)$ is a periodic vertex. Consider the following two cases.

\emph{Case 1.} $r\neq k$. Define
$$
\sigma:=\min\left\{\big||r-k|-(n-1)|\kappa-k|\big|: \kappa\in \mathrm{{supp}}_G(v)\setminus\{r\}\right\}.
$$
Assume that $\theta$ is an eigenvalue achieving the above minimum. Define
$$
n_r:=| r-k|, \text{~and~} n_\theta:=|\theta-k|,
$$
and suppose that $\sigma:=|n_r-(n-1)n_\theta|$.
By Lemma \ref{mainlemma},
$n_\theta^2+4m$ and $n_r^2+4m(n-1)^2$ are squares.  Let
\begin{equation*}
s:=\sqrt{n_\theta^2+4m}, \text{~and~}  t:=\sqrt{n_r^2+4m(n-1)^2}.
\end{equation*}
Then
\begin{equation*}
(n-1)s+t>(n-1)n_\theta+n_r, \text{~and~} |t^2-((n-1)s)^2|=((n-1)n_\theta+n_r)\sigma,
\end{equation*}
which implies $|t-(n-1)s|<\sigma$. By (\ref{mainE2}), we have $\sigma=1,2$. If $\sigma=1$, then $|t-(n-1)s|<1$, which cannot occur.
If $\sigma=2$, then $|t-(n-1)s|<2$, which contradicts that $t$ and $(n-1)s$ have the same parity.

\emph{Case 2.} $r=k$.
Note that $\sqrt{4m(n-1)^2}$ is an integer multiple of $\sqrt{\Delta}$. Define
$$
\sigma:=\frac{1}{\sqrt{\Delta}}\min\left\{\big|(n-1)|\kappa-k|\big|: \kappa\in \mathrm{{supp}}_G(v)\setminus\{r\}\right\}.
$$
Assume that $\theta$ is an eigenvalue achieving the above minimum. Define
$$
n_\theta:=\frac{|\theta-k|}{\sqrt{\Delta}},
$$
and suppose that  $\sigma:=(n-1)n_\theta$.
By Lemma \ref{mainlemma},
$n_\theta^2+4m/\Delta$ is a square. Let
\begin{equation*}
s:=\sqrt{n_\theta^2+\frac{4m}{\Delta}},\text{~and~}  t:=\sqrt{\frac{4m(n-1)^2}{\Delta}}.
\end{equation*}
Then
\begin{equation*}
(n-1)s+t>\sigma, \text{~and~}  |t^2-((n-1)s)^2|= \sigma^2,
\end{equation*}
which implies $|t-(n-1)s|<\sigma$. By (\ref{mainE2}), we have $\sigma=1,2$. If $\sigma=1$, then $|t-(n-1)s|<1$, which cannot occur.
If $\sigma=2$, then $|t-(n-1)s|<2$, which contradicts that $t$ and $(n-1)s$ have the same parity.

This completes the proof.
\qed\end{proof}

\begin{cor}\label{pstcor}
Let $G$ and $\overrightarrow{H}$ be as in Theorem \ref{eigenprojector}, and let $v$ be a vertex of $G$.
\begin{itemize}
\item[\rm (a)] If there are two distinct eigenvalues $\lambda,\mu\in \mathrm{{supp}}_G(v)\setminus\{r\}$ such that
    \begin{equation}\label{corE1}
    0<|\lambda-k|-|\mu-k|<3,
    \end{equation}
    then $(v,w)$ is not periodic in $G\tilde\circ\overrightarrow{H}$, for all $w\in V(H_i)\cup\{0\}$.
\item[\rm (b)] If there is an eigenvalue $\kappa\in \mathrm{{supp}}_G(v)\setminus\{r\}$ such that
\begin{equation}\label{corE2}
0<\big|| r-k|-(n-1)|\kappa-k |\big|<3,
\end{equation}
then $(v,w)$ is not periodic in $G\tilde\circ\overrightarrow{H}$, for all $w\in V(H_i)\cup\{0\}$.
\end{itemize}
\end{cor}

\begin{proof}
(a) By contradiction, suppose that $(v,0)$ is a periodic vertex. By Lemma \ref{mainlemma}, there exists a square-free integer $\Delta\geq1$ such that both $\lambda-k$ and $\mu-k$ are integer multiples of $\sqrt{\Delta}$.
By (\ref{corE1}), we have
$$
|\lambda-k|-|\mu-k|\in\left\{\sqrt{1}, \sqrt{2}, \sqrt{3}, 2\sqrt{1}, \sqrt{5}, \sqrt{6}, \sqrt{7}, 2\sqrt{2}\right\}.
$$
This contradicts to Theorem \ref{miantheorem} (a).

(b) By contradiction, suppose that $(v,0)$ is a periodic vertex. Consider the following two cases.

 \emph{Case 1.} $r\neq k$.  By Lemma \ref{mainlemma} (a), both $\kappa-k$ and $r-k$ are integers. By (\ref{corE2}), we have
 $$
 \big|| r-k|-(n-1)|\kappa-k |\big|\in\left\{\sqrt{1}, 2\sqrt{1}\right\}.
 $$
This contradicts to Theorem \ref{miantheorem} (b).

\emph{Case 2.} $r=k$. By Lemma \ref{mainlemma} (b), $\kappa-k$ is an integer multiples $\sqrt{\Delta}$. By (\ref{corE2}), we have
$$
\big|(n-1)|\kappa-k|\big|\in\left\{\sqrt{1}, \sqrt{2}, \sqrt{3}, 2\sqrt{1}, \sqrt{5}, \sqrt{6}, \sqrt{7}, 2\sqrt{2}\right\}.
$$
This also contradicts Theorem \ref{miantheorem} (b).
\qed\end{proof}

\begin{example}\label{ex1}
{\em Let $G$ be the $d$-dimensional cube with $d\geq2$. Then the set of all distinct eigenvalues of $G$ is $\mathrm{Spec}_G=\{d-2l: 0\leq l\leq d\}$ \cite[Theorem 9.2.1]{BrouwerCN89}. Note that $G$ is a distance-regular graph. Then $\mathrm{Spec}_G$ is contained in the eigenvalue support of every vertex of $G$ \cite[Page~41]{Coh14}. In particular, $2-d$ and $-d$ are always eigenvalues of $G$. Therefore, for an arbitrarily $k$,
$$0<|-d-k|-|2-d-k|<3,$$
which satisfies the condition of the Corollary \ref{pstcor} (a). Hence, for an arbitrary $k$-regular graph $H$, every vertex of $G\tilde\circ H$ is not periodic. Moreover, by Lemma \ref{periodic-pst}, $G\tilde\circ H$ has no PST.
}
\end{example}

\subsection{PGST in vertex complemented coronas}

In this section, we prove that vertex complemented coronas have PGST. Before proceeding, we give the following result.

\begin{theorem}\label{PGSTVCCro-11}
Let $G$ and $\overrightarrow{H}$ be as in Theorem \ref{eigenprojector}, and let $u$ and $v$ be two distinct vertices of $G$. For each eigenvalue $\lambda\neq r$ of $G$, define $\Lambda_\lambda=\sqrt{(\lambda-k)^2+4m}$ and $\Lambda_r=\sqrt{(r-k)^2+4m(n-1)^2}$. Then
\begin{align*}
\mathbf{e}_{(u,0)}e^{-\mathrm{i}t A_{G\tilde\circ \overrightarrow{H}}}\mathbf{e}_{(v,0)}
=&\sum_{\lambda\in \mathrm{{Spec}}_G\setminus\left\{r\right\}}e^{-\mathrm{i}t(\lambda+k)/2}\left(\cos\left(\frac{\Lambda_\lambda t}{2}\right)-\mathrm{i}\frac{\lambda-k}{\Lambda_\lambda}\sin\left(\frac{\Lambda_\lambda t}{2}\right)\right)\mathbf{e}_u^\top E_\lambda(G)\mathbf{e}_v\\
&+e^{-\mathrm{i}t(r+k)/2}\left(\cos\left(\frac{\Lambda_r t}{2}\right)-\mathrm{i}\frac{r-k}{\Lambda_r}\sin\left(\frac{\Lambda_r t}{2}\right)\right)\mathbf{e}_u^\top E_r(G)\mathbf{e}_v.
\end{align*}
\end{theorem}

\begin{proof}
Recall that $\lambda_\pm=\frac{1}{2}(\lambda+k\pm\Lambda_\lambda)$ for $\lambda\neq r$ and $r_\pm=\frac{1}{2}(r+k\pm\Lambda_r)$. By Theorem \ref{theorem-eigenprojector} and Equation (\ref{SpecDec2-1}), we have
 \begin{align}\label{PGSTE1}
\mathbf{e}^\top_{(u,0)}\mathbf{e}^{-\mathrm{i}t A_{G\tilde\circ \overrightarrow{H}}}\mathbf{e}_{(v,0)}
=&\sum_{\lambda\in \mathrm{{Spec}}_G\setminus\left\{r\right\}} e^{-\mathrm{i}t\frac{\lambda+k}{2}}
\mathbf{e}_u^\top E_\lambda(G)\mathbf{e}_v
\left(\sum_{\pm}
e^{\mp\mathrm{i}t\frac{\Lambda_\lambda}{2}} \frac{(\lambda_{\pm}-k)^2}{(\lambda_{\pm}-k)^2+m}\right)\nonumber\\
&+e^{-\mathrm{i}t\frac{r+k}{2}}
\mathbf{e}_u^\top E_r(G)\mathbf{e}_v
\left(\sum_{\pm}
e^{\mp\mathrm{i}t\frac{\Lambda_r }{2}}\frac{(r_{\pm}-k)^2}{(r_{\pm}-k)^2+m(n-1)^2}\right).
 \end{align}
By Maple, we have
\begin{equation}\label{PGSTE2}
\sum_{\pm}e^{\mp\mathrm{i}t\frac{\Lambda_\lambda }{2}}\frac{(\lambda_{\pm}-k)^2}{(\lambda_{\pm}-k)^2+m} =\cos\left(\frac{\Lambda_\lambda t}{2}\right)-\mathrm{i}\frac{\lambda-k}{\Lambda_\lambda}\sin\left(\frac{\Lambda_\lambda t}{2}\right),
\end{equation}
and
\begin{align}\label{PGSTE3}
\sum_{\pm}
e^{\mp\mathrm{i}t\frac{\Lambda_r }{2}}\frac{(r_{\pm}-k)^2}{(r_{\pm}-k)^2+m(n-1)^2} =\cos\left(\frac{\Lambda_r t}{2}\right)-\mathrm{i}\frac{r-k}{\Lambda_r}\sin\left(\frac{\Lambda_r t}{2}\right).
\end{align}
Plugging (\ref{PGSTE2}) and (\ref{PGSTE3}) into (\ref{PGSTE1}), we obtain the required result.
\qed\end{proof}

Let $G$ be a regular connected graph. From Corollary \ref{kmpst}, we know that $G\tilde\circ K_m$ has no PST. In contrast, we use Theorem \ref{PGSTVCCro-11} to prove that $G\tilde\circ K_1$ has PGST.

\begin{theorem}\label{pgst1}
Let $G$ be an $r$-regular connected graph with $n\ge2$ vertices and let $u$, $v$ be two distinct vertices of $G$. If there exists PST from $u$ to $v$ at time $t=\pi/g$, for some positive integer $g$, $0\notin \mathrm{{supp}}_G(u)$ and $r^2+4(n-1)^2$ is not a perfect square, then there exists PGST from $(u,0)$ to $(v,0)$ in $G\tilde\circ K_1$.
\end{theorem}

\begin{proof}
Note that there exists PST from $u$ to $v$ at time $t=\pi/g$ in $G$, for some integer $g$. According to the last sentence of Lemma \ref{pst}, we have $\Delta=1$, that is, all eigenvalues in $\mathrm{{supp}}_G(u)$ are integers.  Note that $r$ is always in $\mathrm{{supp}}_G(u)$. For each eigenvalue $\lambda\in \mathrm{{supp}}_G(u)\setminus\{r\}$, let $c_\lambda$ be the square-free part of $\lambda^2+4$. Then
$$
\Lambda_\lambda=\sqrt{\lambda^2+4}=s_\lambda\sqrt{c_\lambda}
$$
for some integer $s_\lambda$. Note that $0\notin \mathrm{{supp}}_G(u)$. Then $\Lambda_\lambda$ is irrational and $c_\lambda>1$ for each $\lambda\in \mathrm{{supp}}_G(u)\setminus\{r\}$.

Notice that $r^2+4(n-1)^2$ is not a perfect square. Then $\Lambda_r=\sqrt{r^2+4(n-1)^2}$ is irrational. Let $c_r$ be the square-free part of $r^2+4(n-1)^2$. Then $\Lambda_r=s_r\sqrt{c_r}$ for some integer $s_r$.

By Corollary \ref{independent},
 $$
 \left\{\sqrt{c_\lambda}:\lambda\in \mathrm{{supp}}_G(u)\right\}\cup\{1\}
 $$
 is linearly independent over $\mathbb{Q}$. By Theorem \ref{hardy-wright}, there exist integers $l$, $q_\lambda$ such that
 \begin{equation}\label{approx1}
l\sqrt{c_\lambda}-q_\lambda\approx-\frac{\sqrt{c_\lambda}}{2g} \text{~~for~}\lambda\in \mathrm{{supp}}_G(u).
\end{equation}
 Multiplying both sides of (\ref{approx1}) by $4s_\lambda$, we have
 \begin{equation*}
 \left(4l+\frac{2}{g}\right)\Lambda_\lambda\approx4q_\lambda s_\lambda\text{~~for~}\lambda\in \mathrm{{supp}}_G(u).
 \end{equation*}
In particular,
\begin{equation*}
 \left(4l+\frac{2}{g}\right)\Lambda_r\approx4q_r s_r.
\end{equation*}
Hence, let $T=(4l+2/g)\pi$, we have $\cos(\Lambda_\lambda T/2)\approx1$ for $\lambda\in \mathrm{{supp}}_G(u)$.
By Theorem \ref{PGSTVCCro-11},
\begin{align*}
\mathbf{e}_{(u,0)}e^{-\mathrm{i}T A_{G\tilde\circ K_1}}\mathbf{e}_{(v,0)}
=&\sum_{\lambda\in \mathrm{{Spec}}_G\setminus\left\{r\right\}}e^{-\mathrm{i}T\lambda/2} \left(\cos\left(\frac{\Lambda_\lambda T}{2}\right)-\mathrm{i}\frac{\lambda}{\Lambda_\lambda} \sin\left(\frac{\Lambda_\lambda T}{2}\right)\right)\mathbf{e}_u^\top E_\lambda(G)\mathbf{e}_v\\
&+e^{-\mathrm{i}Tr/2}\left(\cos\left(\frac{\Lambda_r T}{2}\right)-\mathrm{i}\frac{r}{\Lambda_r}\sin\left(\frac{\Lambda_r T}{2}\right)\right)\mathbf{e}_u^\top E_r(G)\mathbf{e}_v\\
\approx&\sum_{\lambda\in \mathrm{{Spec}}_G}e^{-\mathrm{i}(2\pi)l\lambda} e^{-\mathrm{i}\lambda\pi/g}\mathbf{e}_u^\top E_\lambda(G)\mathbf{e}_v\\
=&\mathbf{e}_u^\top e^{-\mathrm{i}(\pi/g)A_G}\mathbf{e}_v.
\end{align*}
Note that $G$ has PST from $u$ to $v$ at time $\pi/g$. Then $|\mathbf{e}_u^\top e^{-\mathrm{i}(\pi/g)A_G}\mathbf{e}_v|=1$. Therefore, $|\mathbf{e}_{(u,0)}e^{-\mathrm{i}T A_{G\tilde\circ K_1}}\mathbf{e}_{(v,0)}|\approx1$, that is, there exists PGST from $(u,0)$ to $(v,0)$ in $G\tilde\circ K_1$.
\qed\end{proof}

\begin{example}
{\em Let $G$ be the double coset graph of binary Golay code \cite[Page~415]{BrouwerCN89}. By Corollary \ref{kmpst}, $G\tilde\circ K_1$ has no PST. Let $u$, $v$ be two distinct vertices of $G$, the set of all distinct eigenvalues of $G$ is $\mathrm{Spec}_G=\{23, 9, 7, 1, -1, -7, -9, -23\}$ and $G$ has PST from $u$ to $v$ at time $\pi/2$ \cite[Page~122]{Coutinho15}. Note that $0\notin \mathrm{{supp}}_G(u)$ and the number of vertices $n=4096$. Then $23^2+4(4096-1)^2=67076629$ is not a perfect square. So by Theorem \ref{pgst1}, there exists PGST from $(u,0)$ to $(v,0)$ in $G\tilde\circ K_1$.

}

\end{example}

In Theorem \ref{pgst1}, $0$ is restricted in the eigenvalue support of $u$. However, if $0\in \mathrm{{supp}}_G(u)$, we need a stronger condition to get PGST in $G\tilde\circ K_1$.

\begin{theorem}\label{pgst2}
Let $G$ be an $r$-regular connected graph with $n\ge2$ vertices and let $u$, $v$ be two distinct vertices of $G$. If $G$ has PST from $u$ to $v$ at time $t=\pi/2$, $0\in \mathrm{{supp}}_G(u)$ and $r^2+4(n-1)^2$ is not a perfect square, then there exists  PGST from $(u,0)$ to $(v,0)$ in $G\tilde\circ K_1$.
\end{theorem}

\begin{proof}
Note that there exists PST from $u$ to $v$ at time $t=\pi/2$ in $G$. By Lemma \ref{pst}, all eigenvalues in $\mathrm{{supp}}_G(u)$ are integers.  Note that $r$ is always in $\mathrm{{supp}}_G(u)$. Then for each eigenvalue $\lambda\in \mathrm{{supp}}_G(u)\setminus\{ r\}$, let $c_\lambda$ be the square-free part of $\lambda^2+4$. Then
$$
\Lambda_\lambda=\sqrt{\lambda^2+4}=s_\lambda\sqrt{c_\lambda}
$$
for some integer $s_\lambda$. Note that  $\Lambda_\lambda$ is irrational and $c_\lambda>1$ for each $\lambda\in \mathrm{{supp}}_G(u)\setminus\{0, r\}$ and $c_\lambda=1$ if and only if $\lambda=0$.

Notice that $r^2+4(n-1)^2$ is not a perfect square. Then $\Lambda_r=\sqrt{r^2+4(n-1)^2}$ is irrational. Let $c_r$ be the square-free part of $r^2+4(n-1)^2$. Then $\Lambda_r=s_r\sqrt{c_r}$ for some integer $s_r$.

By Corollary \ref{independent},
$$\{\sqrt{c_\lambda}:\lambda\in \mathrm{{supp}}_G(u)\setminus\{0\}\}\cup\{1\}$$
is linearly independent over $\mathbb{Q}$.
By Theorem \ref{hardy-wright}, there exist integers $l$, $q_\lambda$ such that
 \begin{equation}\label{approx2}
l\sqrt{c_\lambda}-q_\lambda\approx-\frac{\sqrt{c_\lambda}}{4}+\frac{1}{2s_\lambda} \text{~~for~} \lambda\in \mathrm{{supp}}_G(u)\setminus\{ 0\}.
\end{equation}
 Multiplying both sides of (\ref{approx2}) by $4s_\lambda$, we have
 \begin{equation*}
 (4l+1)\Lambda_\lambda\approx4q_\lambda s_\lambda+2 \text{~~for~} \lambda\in \mathrm{{supp}}_G(u)\setminus\{ 0\}.
 \end{equation*}
Hence, let $T=(4l+1)\pi$, we have $\cos(\Lambda_0 T/2)=-1$ and $\cos(\Lambda_{\lambda} T/2)\approx-1$ for $\lambda\in \mathrm{{supp}}_G(u)\setminus\{ 0\}$. By Theorem \ref{PGSTVCCro-11},
\begin{align*}
\mathbf{e}_{(u,0)}e^{-\mathrm{i}T A_{G\tilde\circ K_1}}\mathbf{e}_{(v,0)}
=&\sum_{\lambda\in \mathrm{{Spec}}_G\setminus\left\{r\right\}}e^{-\mathrm{i}T\lambda/2} \left(\cos\left(\frac{\Lambda_\lambda T}{2}\right)-\mathrm{i}\frac{\lambda}{\Lambda_\lambda} \sin\left(\frac{\Lambda_\lambda T}{2}\right)\right)\mathbf{e}_u^\top E_\lambda(G)\mathbf{e}_v\\
&+e^{-\mathrm{i}Tr/2}\left(\cos\left(\frac{\Lambda_r T}{2}\right)-\mathrm{i}\frac{r}{\Lambda_r}\sin\left(\frac{\Lambda_r T}{2}\right)\right)\mathbf{e}_u^\top E_r(G)\mathbf{e}_v\\
\approx&-\sum_{\lambda\in \mathrm{{Spec}}_G}e^{-\mathrm{i}(2\pi)l\lambda} e^{-\mathrm{i}\lambda\pi/2}\mathbf{e}_u^\top E_\lambda(G)\mathbf{e}_v\\
=&-\mathbf{e}_u^\top e^{-\mathrm{i}(\pi/2)A_G}\mathbf{e}_v.
\end{align*}
 Note that  $G$ has PST from $u$ to $v$ at time $\pi/2$. Then $|\mathbf{e}_u^\top e^{-\mathrm{i}(\pi/2)A_G}\mathbf{e}_v|=1$. Therefore, $|\mathbf{e}_{(u,0)}e^{-\mathrm{i}T A_{G\tilde\circ K_1}}\mathbf{e}_{(v,0)}|\approx1$, that is, there is PGST between $(u,0)$ and $(v,0)$ in $G\tilde\circ K_1$.
\qed\end{proof}

\begin{example}
{\em Let $G$ be the  coset graph of the shortened binary Golay code \cite[Page~416]{BrouwerCN89} and let $u$, $v$ be two distinct vertices of $G$. The set of all distinct eigenvalues of $G$ is $\mathrm{Spec}_G=\{22, 8, 6, 0, -2, -8, -10\}$ and $G$ has PST from $u$ to $v$ at time $\pi/2$ \cite[Page~122]{Coutinho15}. Note that $G$ is a distance-regular graph. Then $\mathrm{Spec}_G$ is contained in the eigenvalue support of every vertex of $G$ \cite[Page~41]{Coh14}, that is, $0\in \mathrm{{supp}}_G(u)$.  Since the number of vertices $n=2048$, then $22^2+4(2048-1)^2=16761320$ is not a perfect square. So by Theorem \ref{pgst2}, there exists PGST from $(u,0)$ to $(v,0)$ in $G\tilde\circ K_1$.


}
\end{example}

%
%

\end{document}